\documentclass[conference,10pt]{IEEEtran}
\usepackage{array,color,graphicx,amsmath,amssymb,amsthm,epsfig,mathrsfs,epstopdf}
\usepackage{amsthm}
\newtheorem{theorem}{Theorem}
\newtheorem{lemma}{Lemma}

\newtheorem{remarks}{Remarks}[section]
\usepackage{multicol}

\begin{document}
\title{A Framework for Quality of Service with a Multiple Access Strategy}
\author{
\IEEEauthorblockN{Easwar Vivek Mangipudi, Venkatesh Ramaiyan}
\IEEEauthorblockA{Department of Electrical Engineering\\
Indian Institute of Technology Madras\\
Chennai 600036, India\\
Email: easwar.vivek@gmail.com, rvenkat@ee.iitm.ac.in}
}
\maketitle

\begin{abstract}
We study a problem of scheduling real-time traffic with hard delay constraints
in an unreliable wireless channel.
Packets arrive at a constant rate to the network and have to be delivered within a fixed
 number of slots in a fading wireless channel.
For an infrastructure mode of traffic with a centralized scheduler,
we are interested in the long time average throughput achievable
for the real time traffic.
In \cite{prk1}, the authors have studied the feasible throughput vectors by
identifying the necessary and sufficient conditions using work load characterization.
In our work, we provide a characterization of the feasible throughput vectors
using the notion of the rate region.
We then discuss an extension to the network model studied in \cite{prk1} by allowing multiple access during contention and propose an enhancement to the rate region of the wireless network. We characterize the feasible throughput vectors with the multiple access technique and study throughput optimal and utility maximizing strategies for the network scenario.
Using simulations, we evaluate the performance of the proposed strategy and discuss its advantages.
\end{abstract}

\section{Introduction}
Provisioning quality of service for real time traffic in a wireless network is crucial and is gaining importance
with applications of voice and video streaming. They are also critical for sensor network applications
including real time surveillance and applications that involve network control. The critical feature in such applications is that they are delay constrained. While best effort traffic has minimal delay constraints, real time traffic can have strict constraints on delay including hard delay constraints, probabilistic or average delay constraints. In this paper, we consider
an unreliable, time varying wireless network and we study the problem of provisioning resources
for a real time traffic with hard delay constraints.

We consider a network model similar to that reported
in \cite{prk1} to support the QoS and extend it to a multiple access scenario.
In \cite{prk1}, the authors have proposed a framework to deal with delay, throughput and channel reliabilities. They propose a load based characterization and obtain necessary and sufficient conditions for a set of long term average throughput demands of the users to be feasible. They also proposed a simple admission control policy for the system.
We use a similar framework but characterize the system by obtaining the expected rate vectors achievable in each frame (to be defined later). We identify the necessary ergodic schedules for a frame and evaluate the average throughput using these schedules.

We then extend the framework in \cite{prk1} to a multiple access scenario which results in a larger rate region, implying better performance in-terms of throughput. In the extended model, the access point instead of polling the users (as in \cite{prk1}) does a multicast control packet exchange with the users at the beginning of every slot. We obtain the rate region for this extended framework by characterizing the achievable rate vector in each frame. We discuss admission control policies and study throughput and utility optimal schedulers for the extended network. Using simulations, we evaluate the performance of the proposed strategy and comment on its advantages.

\subsection{Literature Survey}
Supporting QoS in wireless networks with emphasis on delay constraints has been an area of great interest in recent years.
In \cite{prk1}, Hou et al, propose a QoS framework for real time traffic that characterizes the feasible throughput region with hard delay constraint. 
In \cite{hou_2010}, the framework was extended to study heterogeneous delay constraints. In our work, we extend the framework studied in \cite{prk1} and \cite{hou_2010} to the multiple access scenario using a rate region viewpoint for the homogeneous delay model. The framework studied in \cite{prk1} and in this work concerns with an infrastructure setup with uplink/downlink traffic and hard delay constraints. References such as \cite{het_delay_srikant}, \cite{neely_superfast} consider generalizations for ad hoc wireless networks.
Scheduling delay constrained packets with modified earliest due date policy has been looked at in works like \cite{shakkottai_2002,elsayed_2006}.
The tradeoff between the average throughput and the average delay of the packets in the network is studied in \cite{neely_superfast}.
In \cite{exponential}, the authors propose an average delay optimal scheduling strategy for multiple flows sharing a time varying channel.
In \cite{fastcsma_eryilmaz}, Li and Eryilmaz propose a fast-CSMA algorithm for deadline constrained scheduling for a distributed wireless channel.
References such as \cite{prk_utility_max} and \cite{hou_utility} study utility maximization strategies for the network.

There are a number of interesting references focusing on QoS provisioning based on current implementations and standards.
In \cite{ganz2003}, Wongthavarawat and Ganz propose a packet scheduling strategy for QoS support in IEEE 802.16 broadband wireless access systems.
In the context of video streaming, packet scheduling algorithm that apply different deadline thresholds to packets has been proposed  in \cite{kang2002}.

\section{Network Model}
We consider an infrastructure wireless network setup with a base station or an access point and a fixed number, $N$, of wireless users. Time is assumed to be slotted and is grouped into frames of $\tau$ slots each as shown in Figure~\ref{fig:network_model}. Each user is assumed to generate a packet of fixed size at the beginning of every frame (uplink traffic scenario). The packets have a strict delay constraint and has to be delivered to the base station within the frame; the users discard the packets at the end of the frame.
\begin{figure}
\begin{center}
\includegraphics[scale=0.45]{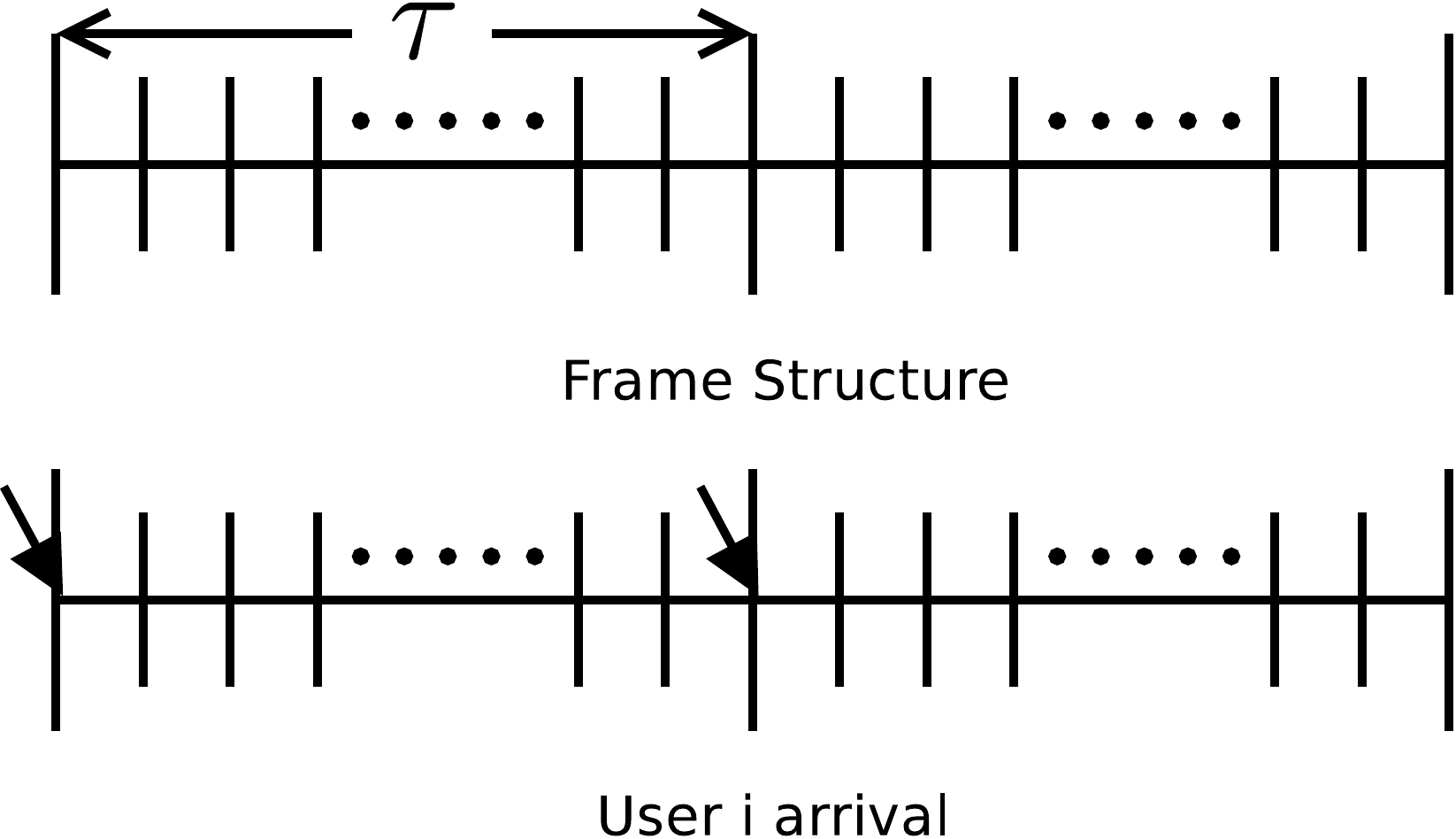}
\caption{Slot structure, frame structure and arrival process for the wireless network.}
\label{fig:network_model}
\end{center}
\end{figure}

We consider an unreliable wireless channel between the access point and the users and the users and the access point (symmetric channel assumption), where the channel is ON with probability $p_i$ for any user $i$ (and is OFF with probability $1-p_i$) in every slot. The channel is assumed to be constant in a time slot and is assumed to be independent across time and users. We assume a centralized scheduling strategy where a single user is successfully scheduled in a slot by the base station or the access point. We assume that the
data packet is transmitted only after successful negotiation initiated by the base station. We assume that the slot duration is long enough to accommodate the transmission of control, data and the acknowledgement packets.
A similar framework was studied in \cite{prk1} where the access point polls a single user in every slot. In this work, we extend the framework to include multiple access during contention (discussed in detail in section IV).

The performance metric of interest in this work is the long time average throughput of a user.  The long time average throughput of user $i$ is defined as
\[ d_i = \liminf_{t \rightarrow \infty} \frac{1}{t} \sum_{k=1}^t I_i(k) \]
where $I_i(t)$ is the indicator function indicating successful transmission for any user $i$ in slot $t$. Let $(d_1, d_2, \cdots, d_N)$ be a feasible long time average throughput vector, i.e., there exists some scheduling strategy,
possibly non-causal, that can achieve this long time average rate vector. Then, the rate region of the wireless network,
${\mathcal C}$ is defined as the set of all feasible throughput vectors $(d_1, d_2, \cdots, d_N)$. The network objective considered
in this work is to identify a scheduler that can support any feasible throughput vector for the wireless network and to maximize
a network utility on the throughput vectors.

\section{Rate Region of the wireless network}
\label{Sec3}
In \cite{prk1}, Hou et al have studied the problem of scheduling users with hard delay requirements
when a single user is polled in a slot. The user is scheduled in the slot if the channel is ON for the polled user.
They characterize the rate region of the wireless network using the load constraints on the network.
In Lemma~5 in \cite{prk1}, the authors show that an N-tuple ($d_1, d_2, \cdots, d_N$) is a feasible throughput vector
if
\[  \sum_{i \in {\mathcal S}} \frac{d_i}{p_i \tau} \leq 1 - {\mathsf E}[I_{{\mathcal S}}] \]
where ${\mathcal S}$ is any subset of $\{ 1,2,\cdots,N \}$ and ${\mathsf E}[I_{{\mathcal S}}]$
is the expected number of idle slots in a frame when the set of ${\mathcal S}$ users are scheduled.
In this work, we like to characterize the rate region directly by identifying the average
rate vector feasible in every frame.

The capacity of a wireless channel
has been studied in a number of works for a variety of network scenarios
(see \cite{maxwt}, \cite{neely_superfast}, \cite{aditya}). We note that the capacity of the delay constrained
wireless network can be similarly defined by identifying that the channel in every frame
is i.i.d. The average rate vector in every frame can be computed by identifying the set of
schedules possible in a frame and by computing the rate vector associated with the schedule.
As the ON probabilities of the channels are identical across slots and frames, the expected throughput obtained for any frame would be same as that of the long term average throughput.


In this work, we propose to consider the following set of schedules for the wireless network. Given any subset
of users $\{ i_1, i_2, \cdots, i_k\}$, consider all possible permutations of the users where a permutation is defined as
$i_1/i_2/\cdots/i_k$. The schedule $i_1/i_2/\cdots/i_k$ corresponds to the strategy of scheduling user $i_1$ until success and then user $i_2$ until success and so on within a frame. We note that there are $k!$ such schedules with a given set of $k$ users and there
are $^{N}C_k$ such combinations for $k$ users. Further, $k$ can be from $0$ to $N$. The following theorem summarizes the sufficiency of such schedules for the wireless network.
\begin{theorem}
The total number of ordered schedules for a network with $N$ users is $\sum_{k=0}^N { }^{N} P_k$. The convex hull of the expected rate vectors of the set of ordered schedules is the set of feasible throughput vectors for the wireless network.
\label{thm:1}
\end{theorem}
\begin{proof}
Here, we will provide an outline of the proof.

The channel and traffic conditions are i.i.d. in every frame. Hence, it is sufficient to restrict to the Ergodic schedules. Consider any ordered schedule from a combination ${\mathcal S} \subset \{1, 2, \cdots, N\}$. Clearly, the throughput achieved using any ordered schedule satisfies the necessary condition $\sum_{i \in {\mathcal S}} \frac{d_i}{\tau p_i} \leq 1 - {\mathsf E}[{\mathcal S}]$. Hence, the convex hull of the throughput vectors achieved using the ordered schedules is a subset of the rate region described in \cite{prk1}. The rate vectors corresponding to the ordered schedules lie in the intersection of such planes described by the necessary conditions in \cite{prk1}. The face of the rate region described by the plane $\sum_{i \in {\mathcal S}} \frac{d_i}{\tau p_i} \leq 1 - {\mathsf E}[{\mathcal S}]$ is constrained by the rate vectors achievable by the ordered schedules that are permutations of ${\mathcal S}$. Hence, the convex hull of the rate vectors achieved through the ordered schedule is in fact the feasible throughput vectors described in \cite{prk1}.
\end{proof}

\begin{remarks}
These set of points described in the Theorem~\ref{thm:1} are the different corner points of the rate region which is now a $N$ dimensional polymatroid (\cite{polymatroid}).
We compute the expected rate vector obtained for a schedule in a frame similar to that computed for a slot in \cite{neely_delay} to obtain the average throughput vector. The rate region has $\sum_{k=0}^N \ ^{N}P_k$ corner points (or schedules). Each of the necessary conditions discussed in \cite{prk1} corresponds to a face on the rate region and hence, there are $N + \sum_{k=1}^N \ ^{N}C_r$ faces on the rate region (including the non-negative constraints for the average throughput vector).
\end{remarks}



We will now illustrate the idea using an example with two users.
The rate region of a two user network $N = 2$ can be obtained by considering the following Ergodic schedules possible in a frame: allocate all the slots to user one (indicated as $1$), to user two ($2$), allocate slots in a frame to user one till the user succeeds and then to user two ($1/2$) and vice versa ($2/1$). In Figure~\ref{fig_rr_comp_2usr}, we have plotted the throughput vectors corresponding to the four schedules $\{ 1,2,1/2,2/1 \}$. The convex hull of the above throughput vectors along with the $(0,0)$ throughput vector is the rate region of the wireless network.


\begin{figure}
\begin{center}
\includegraphics[scale=0.54]{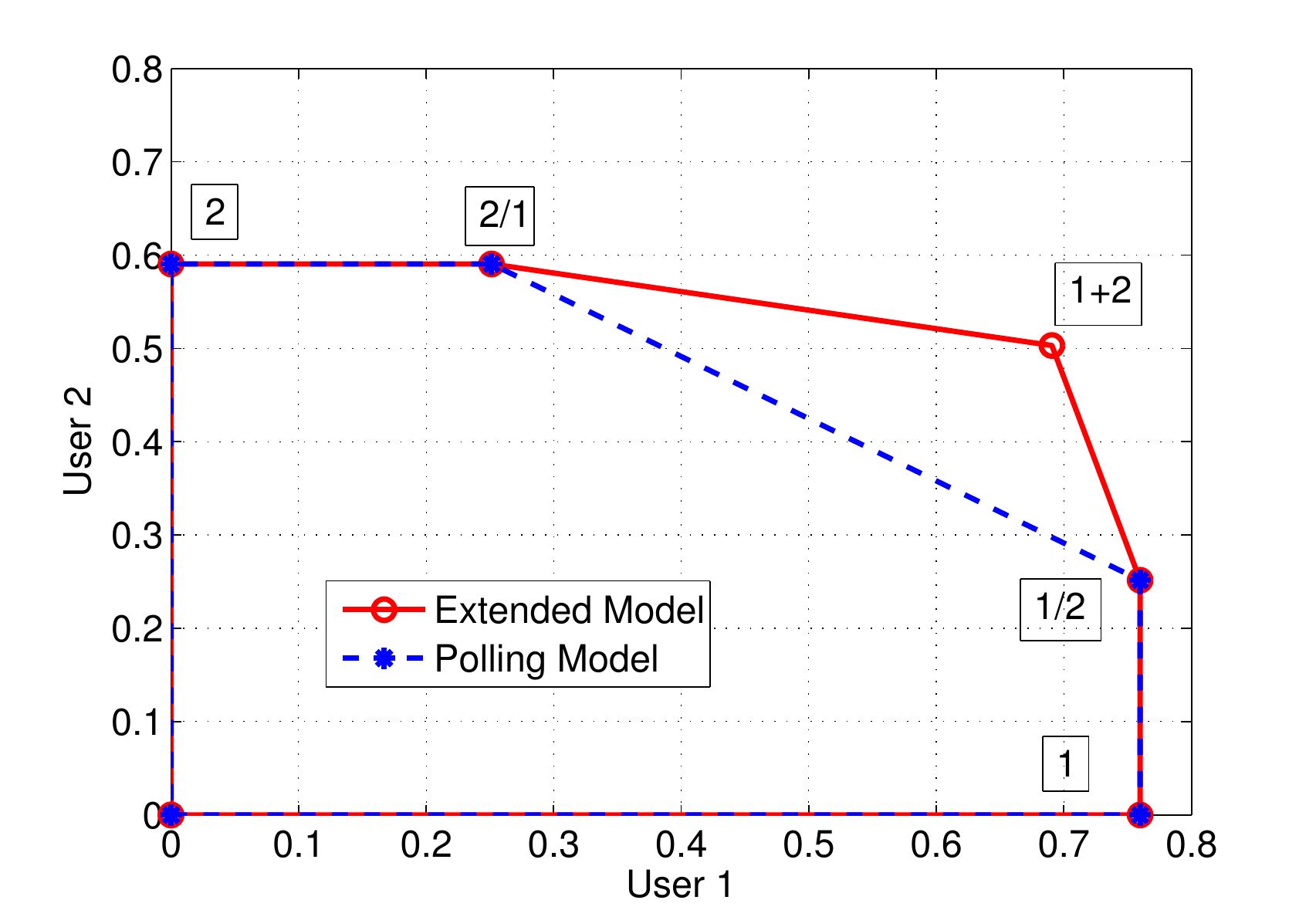}
\caption{Rate regions of the polling model (from \cite{prk1}) and the extended models for a wireless network with $N = 2$ users with $p_1=0.3$, $p_2=0.2$ and $\tau=4$ slots. We have indicated the schedule points for the polling model as well as the extended network model.}
\label{fig_rr_comp_2usr}
\end{center}
\end{figure}


\section{Extended network model}
In this section, we will propose an extension to the network model studied in \cite{prk1} by allowing multiple access during the contention phase. In the polling model studied in \cite{prk1}, at the beginning of a slot, the base station polls a single user using a control packet. The user responds with a data packet if the channel is ON and the slot is wasted if the channel is OFF. We recommend that at the beginning of a slot, the base station shall broadcast a multicast control packet to the users (addressing a subset of the users). The users with the ON channel respond to the multicast control packet with a data packet (for the uplink traffic) if the control packet is addressed to the user. If the base station observes a single successful packet, the user which responded in the slot is deemed scheduled in the slot.
%
%
The multicast contention packet permits multiple access in a slot. If the base station notices a collision or if there is no transmission in the slot, the slot is considered wasted.

In \cite{prk1}, the base station addresses a single user in a slot. We propose that the base station shall address a subset of users that needs to be scheduled in a slot instead. We would expect that the subset chosen shall aim to maximize the network objective in the given slot.
We refer to the network model which permits multiple access in contention as the extended network model. Our aim is to characterize the feasible throughput vectors for the extended network model, identify admission control policy and propose stabilizing and utility optimal schedules as well. We note here that the load based characterization (proposed in \cite{prk1}) does not lead to simple characterization of the feasible throughput vectors; however, we can define the feasible throughput vectors using the rate region viewpoint as illustrated in the example below.

\subsection*{Rate Region of a Two User Example}
Consider the two user network $N = 2$ with frame size $\tau = 4$ slots and with channel probabilities $p_1 = 0.3$ and $p = 0.2$ studied in Figure~\ref{fig_rr_comp_2usr}. In Section~\ref{Sec3}, we noted that the set of efficient schedules permitted with a polling model for contention are $\{ 1, 2, 1/2, 2/1 \}$. The multicast model for contention permits additional schedules such as $(1+2)$, where users $1$ and $2$ are scheduled simultaneously in a multicast combination until a success and $(1+2)^c$ where the user(s) that is not scheduled so far in the multicast combination $(1+2)$ competes until a success.

The user $1$ throughput with the schedule $(1+2)/(1+2)^c$ is
\begin{equation}
\begin{split}
&\sum_{k=1}^T \Big[ ((1-p_2)(1-p_1)+p_1p_2)^{k-1}(1-p_2)p_1+ \\ & \sum_{l=1}^{k-1} ((1-p_2)(1-p_1)+p_1p_2)^lp_2(1-p_1)^{k-l}p_1 \Big]
\end{split}
\end{equation}
and the user $2$ throughput with the schedule is
\begin{equation}
\begin{split}
&\sum_{k=1}^T \Big[ ((1-p_2)(1-p_1)+p_1p_2)^{k-1}(1-p_1)p_2+ \\ & \sum_{l=1}^{k-1} ((1-p_2)(1-p_1)+p_1p_2)^lp_1(1-p_2)^{k-l}p_2 \Big] 
\end{split}
\end{equation}
The convex hull of the throughput vectors of the schedules $\{ 1,2,1/2,2/1,(1+2)/(1+2)^c \}$ gives the set of all feasible throughput vectors for the two user example. In Figure~\ref{fig_rr_comp_2usr}, we have plot the throughput with the extended network model illustrating the enhancement in the throughput in comparison with the polling network model.

The following lemma identifies a condition that guarantees improvement in the rate region with the multiple access technique.
\begin{lemma}
In a $N$ user system, a $k$ user multicast combination will enhance the network rate region if $ p_i \leq \frac{1}{k}$ for all the $k$ users.
\end{lemma}
\begin{proof}
It is sufficient to prove for $\tau=1$ as the channel is i.i.d in each slot and frame.


Consider a wireless network with $N$ users and $\tau = 1$. A throughput vector $(d_1,d_2,\cdots,d_N)$ is feasible if it satisfies the following necessary condition.
\begin{equation}
\frac{d_1}{p_1}+\frac{d_2}{p_2}+\cdots+\frac{d_N}{p_N} \leq 1
\end{equation}


Without loss of generality, consider a multicast combination with the first $k$ users. Trivially, for the polling model, we require that $\sum_{i=1}^k \frac{d_i}{p_i} \leq 1$. The throughput vector achieved using multiple access technique for the $k$ users is
\begin{equation}
\big(p_1 \prod_{i\neq 1}(1-p_i),p_2 \prod_{i\neq 2}(1-p_i),.....,p_k \prod_{i\neq k} (1-p_i) \big )
\label{N_user_pt}
\end{equation}
Substituting the throughput vector (\ref{N_user_pt}) in the equation of the plane $\sum_{i=1}^k \frac{d_i}{p_i}$ and letting $p_i \leq \frac{1}{k}$, we get,
\begin{equation}
\sum_{i=1}^k \prod_{j\neq i} (1-p_j) \geq \sum_{i=1}^k \prod_{j \neq i} (1 - \frac{1}{k}) = k (1 - \frac{1}{k})^{k-1} \geq 1
\label{pt_val}
\end{equation}
where the inequality is strict for $k > 2$ or $p_i < \frac{1}{k}$. This implies that the rate achieved with the multiple access strategy is strictly beyond the plane characterizing the rate region of the polling network model.
\end{proof}

\subsection{Admission control and Sub-optimal Rate Region}
We can now define the rate region of the extended network model like we defined in Section~\ref{Sec3} for the polling model.
The channel and the traffic is i.i.d. in every frame, hence, we can restrict our attention to Ergodic schedules.
In Section~\ref{Sec3}, we limited our attention to schedules of the form $i_1/\cdots/i_k$. The extended network model permits such schedules as well and hence, the rate region of the extended network model is at least as big as the rate region of the polling model. Further, the extended network model permits different multicast combination of users as schedules in a slot which also necessitates the use of dynamic schedules adapted with the outcome of a multicast combination.

In Section~\ref{Sec3}, we characterized the rate region of the wireless network (studied in \cite{prk1}), as having $\sum_{i=0}^N  \ ^{N} P_r$ schedules or corner-points (with $N + \sum_{i=1}^N \ ^{N} C_r$ faces). Suppose that the desired throughput vector $(d_1, d_2, \cdots, d_N)$ is ordered $(d_{k_1}, d_{k_2}, \cdots, d_{k_N})$ such that $d_{k_1} \leq d_{k_2} \cdots \leq d_{k_N}$. Then, the number of original schedules (corner-points) with the same ordering is $\sum_{i=0}^N 1 = N + 1$ and the number of original faces with the same ordering is $N + 1$ (every original schedule in the restricted region induces a face in the section of the rate region). Thus, the number of conditions that ensures the feasibility of a given average throughput vector is $N$ ($+1$ assuming the rate is non-negative). In \cite{prk1}, the authors have essentially proposed a simple admission control policy for the feasibility of an average throughput vector based on the idea.
The extended network model, however, has many more schedules and faces in comparison with the polling model. For example, for a $N = 3$ user wireless network, the number of schedules and faces (constraints) for the extended network model is $28 (27+1) $ and $24(21+3)$ (much larger than the polling model). 
Hence, we suggest to consider only a few of the multi-access schedules to improve the rate region and thus reduce the complexity of the schedule and the admission control policy. 

In Figure~\ref{figure_rrsuboptimal}, we plot the set of feasible throughput vectors achieved using a simple suboptimal strategy and compare it with the polling model. We consider a wireless network with $N = 3$ users and with channel probabilities $p_i = 0.2$.
For the example considered in the figure, we had considered all static schedules with multicast combinations along with the simple schedules for the polling model. From the figure, we note that the suboptimal strategy is a reasonable approximation and enhances the rate region as good as the optimal strategy (not reported in the figure).
\begin{figure}
\begin{center}
\includegraphics[scale=0.5]{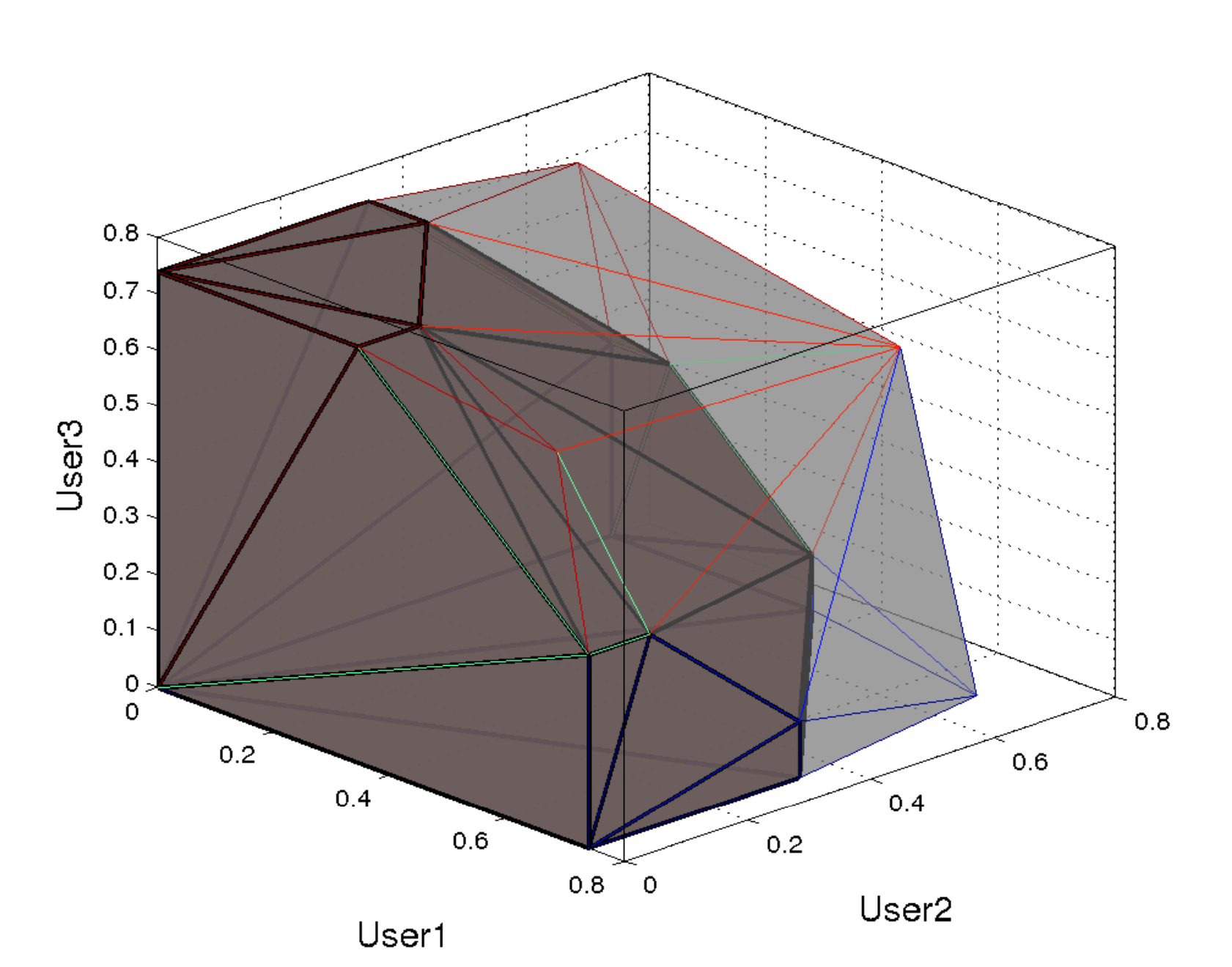}
\caption{The feasible throughput vectors of polling model (darker shade) and a suboptimal extended network model (lighter shade).  We consider a wireless network with $N = 3$ users and channel probabilities $p_i = 0.2$.}
\label{figure_rrsuboptimal}
\end{center}
\end{figure}

The following discussion characterizes the enhancement to the rate region with the extended network model.
\begin{lemma}
The extended network model can enhance the rate region of the wireless network by a factor of $N$.
\end{lemma}
\begin{proof}
Consider the case when $\tau=1$ and the probabilities of ON of all the users is equal to $p$. The system throughput for the polling model in the frame is $p$. But in the case of extended model, the system throughput would be $N p (1-p)^{N-1} \approx N p$ for sufficiently small $p$.
\end{proof}

The extended network model can enhance the network performance even if the number of users that can be combined is limited.
\begin{enumerate}
\item Consider any schedule $i_1/i_2/\cdots/i_k/\cdots/i_m/\cdots/i_N$. The rate region can be enhanced just by combining users $i_k$ and $i_m$ by adding the schedule $i_1/i_2/\cdots/(i_k + i_m)/\cdots/(i_k + i_m)^c/\cdots/i_N$.
\item The improvement in the performance will actually be a function of the actual schedule of operation. The performance enhancement will be the most if the users that are combined are scheduled early in the frame.
\end{enumerate}

\section{Performance Evaluation}
In this section, we report the performance of the wireless network with the multiple access strategy and also compare it with the polling model studied in \cite{prk1}.

\subsection{Throughput and Utility Optimization}
The description of the rate region permits us to propose simple throughput optimal and utility maximizing strategies for the wireless network. In Figure~\ref{figure_maxwt}, by stabilizing the virtual queues (that holds the difference between the number of packets generated and packets transmitted successfully, see \cite{prk1}), we achieve the desired throughput of the wireless users from the extended rate region. In Figure~\ref{figure_PF}, we plot the performance of a gradient scheduler that achieves proportional fairness for the network with the multiple access technique.
\begin{figure}
\begin{center}
\includegraphics[scale=0.5]{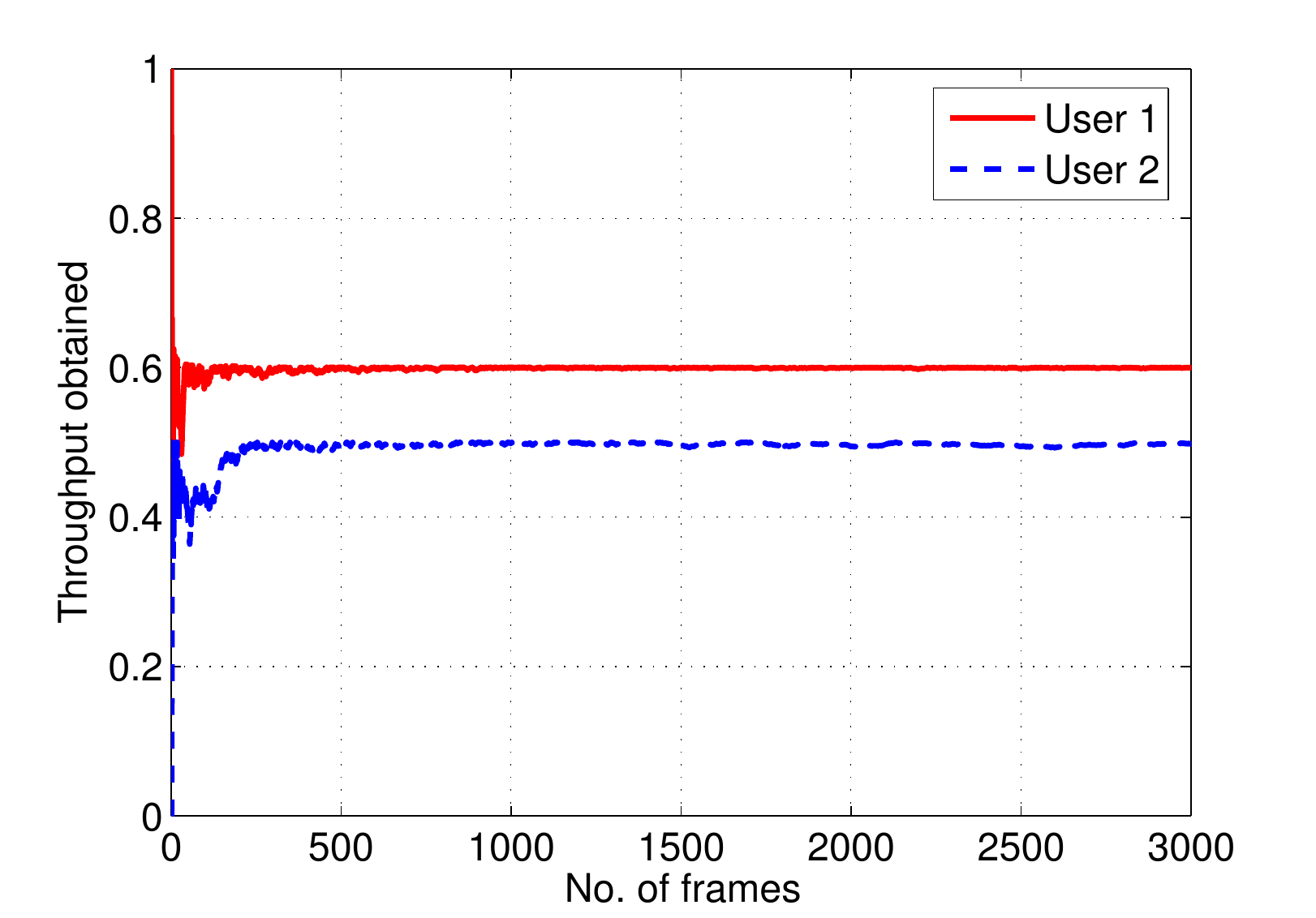}
\caption{Throughput obtained using a virtual queue stabilizing strategy for a wireless network with $N =2$ users, $\tau = 4$ slots and $(p_1,p_2) = (0.3,0.2)$. The minimum throughput requirement is $(0.6,0.5)$ packets per slot. }
\label{figure_maxwt}
\end{center}
\end{figure}
\begin{figure}
\begin{center}
\includegraphics[scale=0.48]{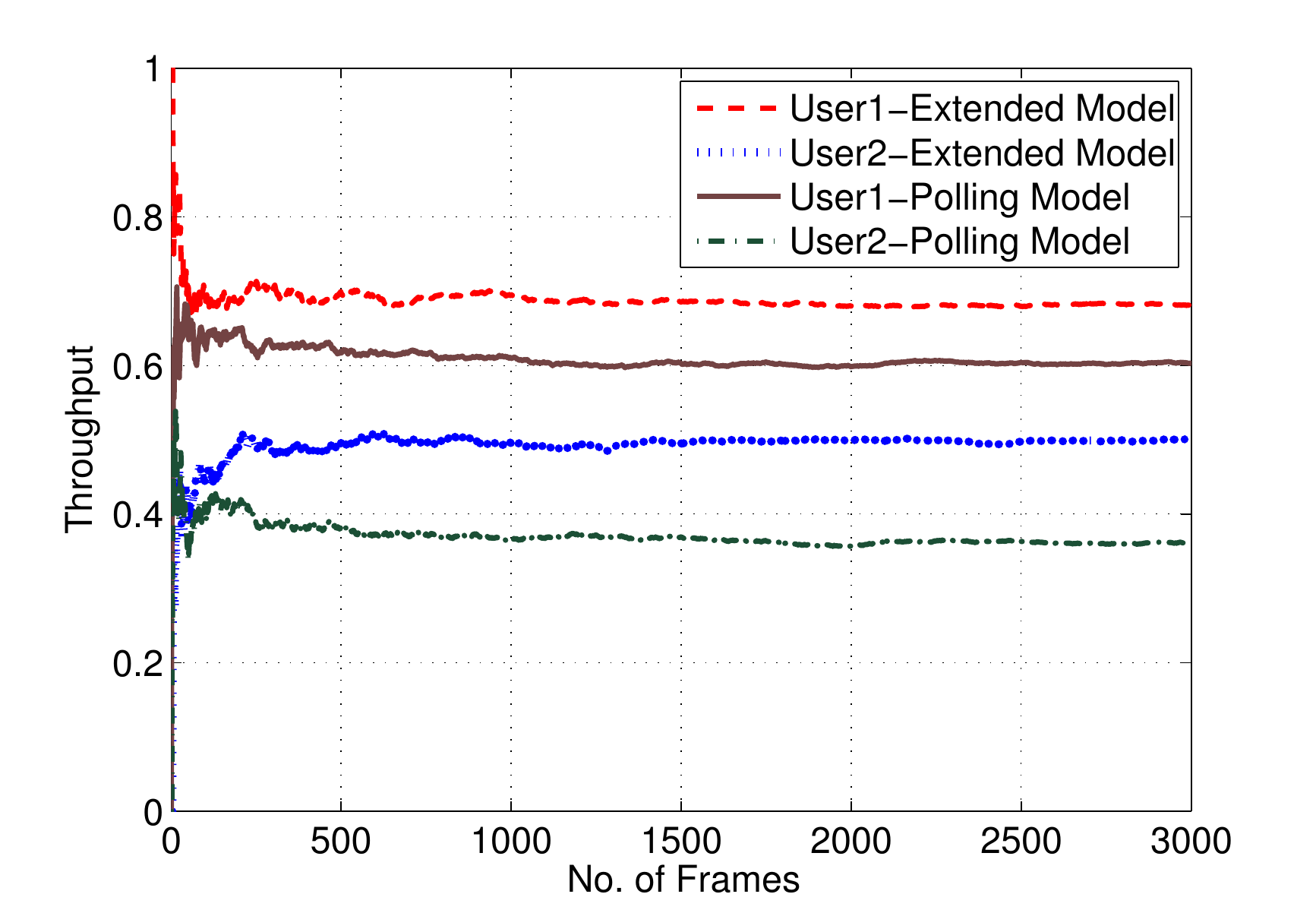}
\caption{Throughput of a proportional fair scheduler implemented using a gradient descent algorithm for a wireless network with $N = 2$ users, $\tau = 4$ slots and $(p_1,p_2) = (0.3,0.2)$. We have also plotted the performance of the proportional fair scheduler for the polling based contention model studied in \cite{prk1}.}
\label{figure_PF}
\end{center}
\end{figure}

\subsection{Cellular Drop}
In Figure~\ref{figure_cdfs}, we report the performance of the multiple access contention model in comparison with the polling model (from \cite{prk1}) for a typical cellular setup. We consider $30$ users deployed randomly and Uniformly in a circle of radius $1$ Km with the base station at the center. The transmit power of the users is fixed at $1$ Watt, the path loss coefficient is $4$ and the SNR threshold value for communication is assumed to be $10$ dBm. We assume that the users experience a Rayleigh fading wireless channel with mean $1$. The frame length is taken as 30 slots. The scheduler seeks to achieve the proportionally fair operating point in the rate region. In Figure~\ref{figure_cdfs}, we plot the CDF of the throughput received by the users for the multiple access scheme and compare it with the polling model. We have considered a suboptimal implementation (with limited number of schedules). We note from the plot that there is considerable improvement in the throughput of the users even with the suboptimal implementation.
\begin{figure}
\begin{center}
\includegraphics[scale=0.53]{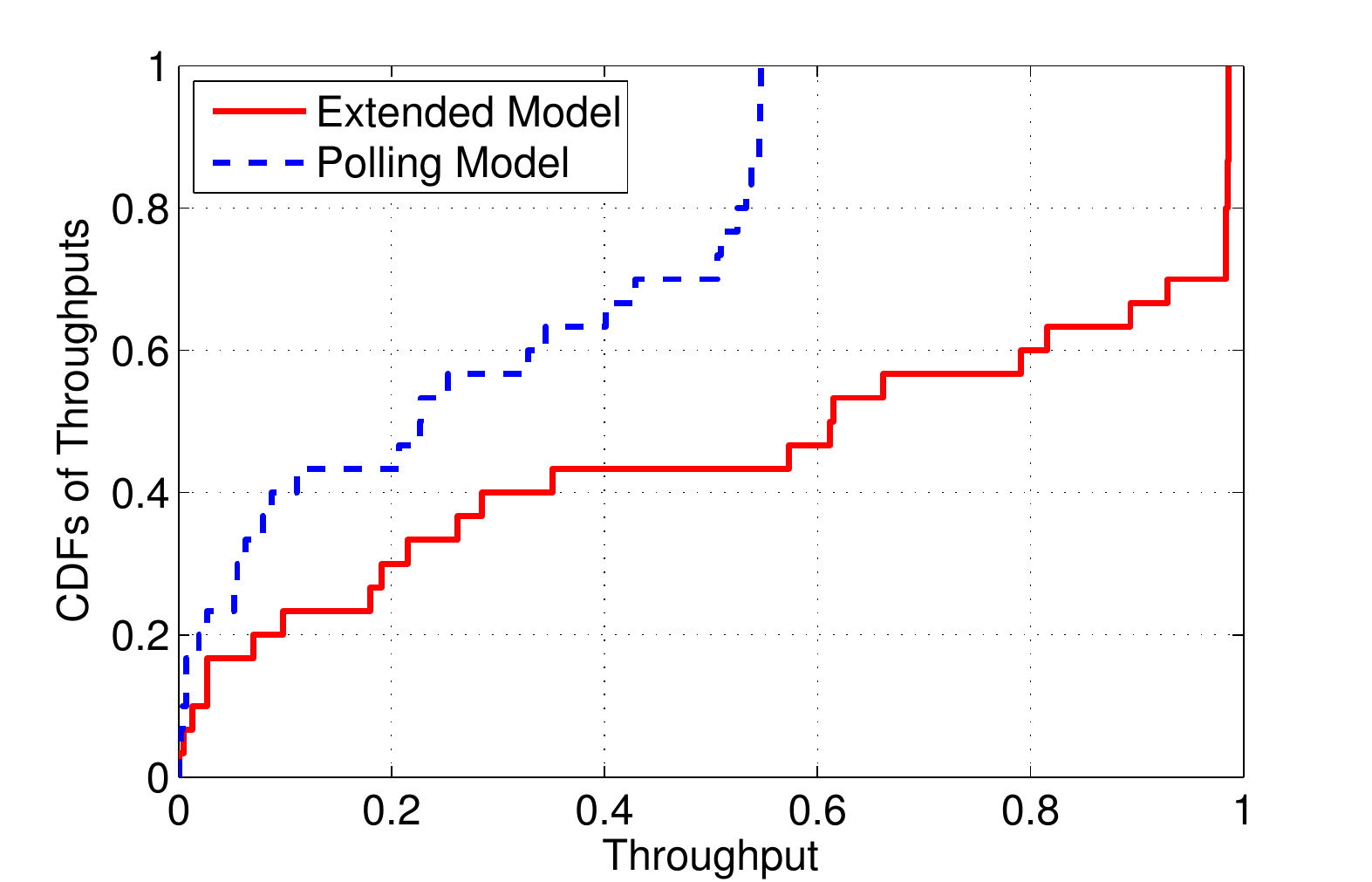}
\caption{CDFs of throughput of the users in a cellular setup for the multiple access contention model and the polling model.}
\label{figure_cdfs}
\end{center}
\end{figure}

\section{Conclusion}
In this work, we have proposed a rate region viewpoint to study a QoS problem in a fading wireless channel with real-time traffic and hard delay constraints. Using multiple access for contention, we report an improvement in the performance of the wireless network in comparison with polling. We characterize the feasible throughput vectors for the wireless network using the rate region viewpoint. We study simple admission control policy and utility maximizing strategies and have evaluated the network performance using simulations.

\bibliography{references}
\bibliographystyle{IEEEtran}

\end{document}